\documentclass{scrartcl}

\usepackage[a4paper,margin=1in, left=7em]{geometry}
\usepackage{amsmath, amssymb}
\usepackage{mathrsfs}
\usepackage{enumitem}
\usepackage[bibliography=common]{apxproof}
\usepackage{stmaryrd}
\usepackage{graphicx}
\usepackage{tikz}
\usepackage{caption, subcaption}
\usepackage{hyperref}

\usepackage[
    type={CC},
    modifier={by},
    version={4.0},
]{doclicense}

\newcommand{\gridstep}{1}
\newcommand{\offsetrad}{.75}
\newcommand{\nodel}[4]{\node[#4] (#1) at (#3) {$\bullet v_{#2}$}}
\newcommand{\nodelc}[2]{\nodel{#1}{#1}{c#1}{#2}}
\newcounter{nodecounter}
\newcommand{\nodeautocol}[1]{\addtocounter{nodecounter}{1} \nodelc{\thenodecounter}{#1}}
\newcommand{\nodeauto}{\addtocounter{nodecounter}{1} \nodelc{\thenodecounter}{}}
\newcommand{\larcrad}{.375}
\newcommand{\barcrad}{.5}
\newcommand{\twohyperedgecol}[4]{\draw[#4] ([yshift=-#3 cm] #1) arc (-90:90:#3) -- ([yshift=#3 cm] #2) arc (90:270:#3) -- cycle}
\newcommand{\twohyperedge}[3]{\twohyperedgecol{#1}{#2}{#3}{}}
\newcommand{\twohyperedgevcol}[4]{\draw[#4] ([xshift=#3 cm] #1) arc (0:180:#3) -- ([xshift=-#3 cm] #2) arc (-180:0:#3) -- cycle}
\newcommand{\twohyperedgev}[3]{\twohyperedgevcol{#1}{#2}{#3}{}}

\newcommand{\trighyperedge}[3]{
  \draw ([shift={(90:-\larcrad)}] #1) arc (-90:90:\larcrad) -- ([shift={(90:\larcrad)}] #2) arc (90:270:\larcrad) -- cycle;
  \draw ([shift={(30:-\larcrad)}] #1) arc (-150:30:\larcrad) -- ([shift={(30:\larcrad)}] #3) arc (30:210:\larcrad) -- cycle;
  \draw ([shift={(-30:-\larcrad)}] #2) arc (150:330:\larcrad) -- ([shift={(-30:\larcrad)}] #3) arc (-30:150:\larcrad) -- cycle
}

\title{Confluence of the Node-Domination and Edge-Domination Hypergraph Rewrite Rules}
\author{Antoine Amarilli$^1$, Mikaël Monet$^1$, Rémi De Pretto$^{1,2}$\\
$^1$: Univ.\ Lille, Inria, CNRS, Centrale Lille, UMR 9189 CRIStAL\\
$^2$: École supérieure de chimie, physique, électronique de Lyon}
\date{}

\usepackage{mathtools}
\usepackage{hyperref}
\hypersetup{
    colorlinks,
    linkcolor={red!50!black},
    citecolor={blue!50!black},
    urlcolor={blue!30!black}
}

\usepackage[capitalise,nameinlink,noabbrev]{cleveref}

\usepackage{algorithmic}
\usepackage[Algorithm]{algorithm}

\usepackage{xcolor}
\definecolor{bordeaux}{rgb}{0.509803922, 0.039215686, 0.039215686}

\usepackage{amsthm}
\theoremstyle{plain}
\newtheorem{theorem}{Theorem}[section]

\theoremstyle{definition}
\newtheorem{definition}[theorem]{Definition}
\newtheorem{proposition}[theorem]{Property}
\newtheorem{corollary}[theorem]{Corollary}
\newtheorem{lemma}[theorem]{Lemma}
\newtheorem{claim}[theorem]{Claim}
\newtheorem{example}[theorem]{Example}

\AddToHook{env/definition/begin}{\crefalias{theorem}{definition}}
\AddToHook{env/proposition/begin}{\crefalias{theorem}{proposition}}
\AddToHook{env/corollary/begin}{\crefalias{theorem}{corollary}}
\AddToHook{env/lemma/begin}{\crefalias{theorem}{lemma}}
\AddToHook{env/claim/begin}{\crefalias{theorem}{claim}}
\AddToHook{env/example/begin}{\crefalias{theorem}{example}}
\crefname{figure}{Figure}{Figures}

\newcommand{\hg}{\mathcal}

\newcommand{\iso}{\equiv}

\newcommand{\chg}[1]{[\mathcal #1]_\iso}
\newcommand{\toe}{\to_{\text{edge}}}
\newcommand{\ton}{\to_{\text{node}}}

\begin{document}

\maketitle

\begin{abstract}
  In this note, we study two rewrite rules on hypergraphs, called
  \emph{edge-domination} and \emph{node-domination}, and show that they are
  confluent. These rules are rather natural and commonly used before computing
  the minimum hitting sets of a hypergraph. Intuitively, edge-domination allows
  us to remove hyperedges that are supersets of another hyperedge, and
  node-domination allows us to remove nodes whose incident hyperedges are a
  subset of that of another node. We show that these rules are
  confluent up to isomorphism, i.e., if we apply any sequences of
  edge-domination and node-domination rules, then the resulting hypergraphs can
  be made isomorphic via more rule applications. This in particular implies the
  existence of a unique \emph{minimal hypergraph}, up to isomorphism.
\end{abstract}

\section{Introduction}

This note studies the confluence of two natural rewrite rules on hypergraphs,
called the \emph{edge-domination} and \emph{node-domination} rules. Informally,
given a hypergraph $\hg H$, the edge-domination rule consists in removing a
hyperedge $e'$ of $\hg H$ if there is another hyperedge $e$ which is contained
in~$e'$, and the node-domination rule consists in removing a node $v$ of $\hg
H$ if there is another node $v'$ which is incident to a superset of the
hyperedges to which $v$ is incident. We give a formal definition of these rules
later in this note.

These two rules have been commonly used for decades, especially when studying
the minimum hitting sets of hypergraphs. Recall that a \emph{hitting set} of a
hypergraph (also called \emph{hypergraph transversal}) is a subset of its nodes
such that every hyperedge contains a node of the hitting set, and a
\emph{minimum hitting set}\footnote{This should not be confused with a
\emph{minimal hitting set}, which refers to minimality under set inclusion, and
which can be tractably computed via a greedy algorithm.} is one of minimal cardinality. It is well-known that
computing minimum hitting sets is NP-hard, because this
includes as a special case the \emph{vertex cover} problem (when the input
hypergraph is a graph, i.e., every hyperedge has cardinality~$2$). The point of
the edge-domination and node-domination rules is that we can apply them to
eliminate some hyperedges and nodes of the input hypergraph in a way  that does
not affect the size of a minimum hitting set. Intuitively this is because
dominated hyperedges will always be covered by the hyperedges that they
contain; and (strictly) dominated nodes will never be part of a vertex cover.
See \cite[Claim~4.8]{biblio:AGMM25} for a formal proof. Thus, the
edge-domination and node-domination rules can be used as an easy
polynomial-time preprocessing on hypergraphs before computing minimum hitting
sets.

In this note, we do not discuss the computation of minimum hitting sets, but
study these rules in their own right. Our goal is to formally show that the
rules are \emph{confluent up to isomorphism}. In other words, consider a hypergraph
$\hg H$. Apply
node-domination and edge-domination rules to $\hg H$ in some fashion to obtain
a hypergraph $\hg H_1$, and apply the rules to $\hg H$ in a different fashion to
obtain a hypergraph $\hg H_2$.
Can we transform $\hg H_1$ and $\hg H_2$ into isomorphic hypergraphs by applying
rules in some fashion on $\hg H_1$ and $\hg H_2$, or is it possible that we
can obtain genuinely different hypergraphs depending on which rules we apply?
Indeed, it is
not difficult to see that~$\hg H_1$ and~$\hg H_2$ may be different~-- e.g., when
we have the choice between removing one of two ``twin'' nodes that are incident to
precisely the same set of hyperedges. However, as we will see, we can
necessarily converge to isomorphic hypergraphs with further rule applications.

In other words, the point of this note is to show that, for any hypergraph $\hg
H$, if we obtain~$\hg H_1$ and~$\hg H_2$ by applying different sequences of rules from $\hg H$,
then we can still converge to isomorphic hypergraphs. 
This fact is rather intuitive, but to our knowledge the question had not been
addressed by previous works. We show:

\begin{theorem}\label{thm:RR confluence informal}
  For any hypergraph $\hg H$, let $\hg H_1$ and $\hg H_2$ be two hypergraphs
  obtained from $\hg H$ by iterative application of edge-domination and
  node-domination rules. Then, there are two hypergraphs $\hg H_1'$ and $\hg
  H_2'$ that are isomorphic and that can be respectively obtained from $\hg H_1$ and $\hg H_2$ by
  application of these rules.
\end{theorem}

In particular, this directly implies that minimal hypergraphs are unique up to
isomorphism, in the following sense:

\begin{corollary}
  \label{cor:maincor}
  Let $\hg H$ be an hypergraph, and let $\hg H_1$ and $\hg H_2$ be two
  hypergraphs obtained from $\hg H$ by iterative application of edge-domination
  and node-domination rules such that it is not possible to apply any rule to
  $\hg H_1$ and $\hg H_2$. Then $\hg H_1$ and $\hg H_2$ are isomorphic.
\end{corollary}

The rest of this note is devoted to a proof of \cref{thm:RR confluence
informal}. It is established using Newman's lemma~\cite{biblio:N42,
biblio:K85}, i.e., by showing local confluence. Specifically, it suffices to
establish the following local confluence claim: for any hypergraph $\hg H$, if
we can obtain $\hg H_1$ and $\hg H_2$ from $\hg H$ in one step (by two
different ways to perform a single rule application), then we can
obtain some hypergraph $\hg H_3$ (up to isomorphism) from $\hg H_1$ and $\hg
H_2$ in a (possibly empty) sequence of steps. The local confluence of the rules
is itself shown by a case distinction, depending on which of the two rules were
applied to obtain $\hg H_1$ and to obtain $\hg H_2$.

\paragraph*{Related work.}
Recall that node-domination and edge-domination have been used for decades
in the study of the minimum hitting
sets of hypergraphs. A review of the use of these rules is presented by Fernau
in~\cite{biblio:Fer10} (and credited to Barretta), with the oldest reference
given being the book of Garfinkel and Nemhauser
\cite[Chapter 8]{biblio:GN72}.
Since then, the rules have been used in many other works: see
\cite{biblio:Wund24} for a recent review.
Some generalisations of the rules have also been proposed, e.g., the
\emph{weighted vertex domination} rule of~\cite[Subsection 2.2]{biblio:Fer10WHS}
generalises node-domination in the setting of weighted hypergraphs.
We note that there are also other rules which can be used as preprocessing for
minimum hitting set problems, such as the \emph{unit hyperedge rule}
of \cite[Section III]{biblio:SC10}: we do not study such rules in this note.
Beyond algorithms for minimum hitting sets, the rules can also be used when
studying hypergraph transversals for other purposes, such as the study of the
resilience problem for databases
in~\cite{biblio:AGMM25}: this is in fact the original motivation of our study
in the present note.

\subparagraph*{Paper structure.} We give preliminaries and give the formal
statement of the problem in
\cref{sec:prelim problem}, rephrasing our informal claim above (\cref{thm:RR
confluence informal}) to a formal claim on confluence
(\cref{thm:RR confluence under isomorphism}). Then, we present the proof of
\cref{thm:RR confluence under isomorphism} in 
\cref{sec:proof}.

\section{Preliminaries}
\label{sec:prelim problem}

\paragraph*{Hypergraphs.}
A \emph{hypergraph} is intuitively defined like a graph but with edges that can be incident to more
than two nodes.
We formalise a hypergraph $\hg H$ as a 3-tuple $(V,E,I)$
where~$V$ is a finite set of \emph{nodes},
$E$ is a finite set of \emph{hyperedges} (or simply \emph{edges}), and 
$I \subseteq V \times E$ is a binary relation called the \emph{incidence relation} :
for $v \in V$ and $e \in E$, we have $v \mathrel{I} e$
(i.e., $(v,e) \in I$) if $e$ is \emph{incident} to $v$.
For $v \in V$ a node, we define~$E(v)$ as the set of edges incident to $v$,
i.e., $E(v) \coloneq \{e \in E \colon v \mathrel{I} e\}$. Symmetrically, for $e
\in E$ an edge, we define $V(e)$
as the set of nodes that $e$ is incident to, i.e., $V(e) \coloneq \{v \in V
\colon v \mathrel{I} e\}$.

Note that our definition of hypergraphs is a generalisation of the classic
definition where the hyperedges of $E$ is a set of subsets of~$V$. Indeed,
unlike that definition, our definition makes it possible to have two distinct hyperedges
containing the same set of nodes. 
For this reason, the hypergraphs that we define here can also be seen as
a notion of \emph{multi-hypergraphs}, i.e., hypergraphs defined according to the
classic definition but with $E$ being
a \emph{multiset}
of subsets of~$V$. However, in the present note, we use the definition with an
incidence relation given in the paragraph above, because it makes the role of
hyperedges and nodes symmetric. 
Moreover, the relation $I$ can be viewed as a matrix $M$, called the
\emph{incidence matrix}, where rows and columns represent nodes and hyperedges, respectively.
For $v$ a node and $e$ a hyperedge, the cell $M[v,e]$ of the matrix~$M$
contains $1$ if $v \mathrel{I} e$ and $0$ otherwise.

\paragraph*{Rewrite rules on hypergraphs.}
A \emph{rewrite rule} on hypergraphs is a process that, given a hypergraph $\hg H$ as input, transforms this
hypergraph into another hypergraph $\hg H'$. Here we present three rewrite rules,
the third rule being the union of the first two.

\begin{definition}[Edge-domination]\label{def:edge-domination}
  Let $\hg H = (V,E,I)$ be a hypergraph. If there are two edges $e,e' \in E$ such that $e \neq e'$ and
  $e'$ is incident to a superset of the nodes to which $e$ is incident,
  then $\hg H' = (V',E',I')$ is obtained by removing $e'$ from $\hg H$.
  Formally, if $e \neq e'$ and $V(e) \subseteq V(e')$ then 
  we set $V' \coloneq V$, $E' \coloneq E \setminus \{e'\}$ and $I' \coloneq \{(v,f) \in I \colon f \neq e'\}$.
  We denote this by $\hg H \toe \hg H'$.
\end{definition}

Notice that this definition is consistent with the usual one where hyperedges
are sets of nodes: the rule then says that we can remove a hyperedge if it is a
superset of another hyperedge. Note that the inclusion is not necessarily
strict, because in our definition of hypergraphs we can have multiple edges
containing the same set of nodes: in this case we can choose any of the edges
and remove it using the rule.

\begin{definition}[Node-domination]\label{def:node-domination}
  Let $\hg H = (V,E,I)$ be a hypergraph. If there are two nodes $v,v' \in V$ such that $v \neq v'$ and
  $v'$ is incident to a superset of the edges to which $v$ is incident,
  then $\hg H' = (V', E')$ is obtained by removing $v$ from $\hg H$.
  Formally, if $E(v) \subseteq E(v')$ then we set
  $V' \coloneq V \setminus \{v\}$, $E' \coloneq E$ and $I' \coloneq \{(u,e) \in I \colon u \neq v\}$.
  We denote this by $\hg H \ton \hg H'$.
\end{definition}

For $\hg H$ a hypergraph, if $\hg H'$ is obtained from~$\hg H$ by some
application of a rewrite rule above,
then we write $\hg H \to \hg H'$ (i.e. ${\to} = {\toe} \cup {\ton}$).
If $\hg H'$ is obtained by applying some (possibly empty) sequence of rewrite rules,
we call $\hg H'$ a \emph{reduced hypergraph} of $\hg H$ and denote it $\hg H \to^* \hg H'$.
In other words, $\to^*$ is the reflexive and transitive closure of $\to$.

Since each rewrite rule eliminates either a
node or a hyperedge, and since we work with finite hypergraphs only, it is clear that any sequence of rule
applications must terminate:

\begin{claim}\label{claim:RR terminates}
  The rewrite rule $\to$ \emph{terminates}, i.e., for every hypergraph $\hg H$, any sequence of
  the form $\hg H = \hg H_0 \to \hg H_1 \to \hg H_2 \to \ldots$ must be finite.
\end{claim}

This motivates the definition of a \emph{minimal} hypergraph:

\begin{definition}[Minimal hypergraph] \label{def:minimal hypergraph}
  For $\hg H$ a hypergraph, if we cannot apply any rewrite rules
  to $\hg H$ (formally if there is no hypergraph $\hg H'$ such that
  $\hg H \to \hg H'$), then we say that $\hg H$ is \emph{minimal}.
  We call $\hg H'$ a \emph{minimal hypergraph of $\hg H$} if $\hg H'$
  is a reduced hypergraph of $\hg H$ which is minimal.
\end{definition}

  By \cref{claim:RR terminates}, one gets that
  every hypergraph has at least one minimal hypergraph. The question is then to
  understand how different choices of applying the 
  rules can lead to different minimal hypergraphs. Note
  that this can happen already in simple cases:

\begin{example}
  \label{exa:nonconfluent}
  Let $\hg H = (V,E,I)$ be defined by taking $V = \{v_1, v_2\}$ and $E$ consisting of
  a single edge $e$ incident to $v_1$ and $v_2$. Then we can apply the node-domination rule to
  remove~$v_1$, yielding $\hg H_1 = (\{v_2\}, \{e\}, \{(v_2,e)\})$; or we can
  apply the same rule to remove~$v_2$,
  yielding $\hg H_2 = (\{v_1\}, \{e\},\{(v_1,e)\})$ (see \cref{fig:twin nodes}).
  The hypergraphs $\hg H_1$ and $\hg H_2$ are both minimal 
  hypergraphs of~$\hg H$, and they are different~-- but they are isomorphic.
  See \cref{fig:confluence iso} for another example.
  We will show that the minimal hypergraphs of any hypergraph are always
  isomorphic.
\end{example}

\begin{example}
  \label{exa:alternating rule}
  Let $\hg H = (V,E,I)$ be a hypergraph with $V = \{v_1, v_2, \dots, v_7\}$ and
  $E = \{e_1, e_2, \dots, e_7\}$ such that for all $i \in \llbracket 2,7 \rrbracket$,
  we have $v_i \mathrel I e_i$, $v_i \mathrel I e_{i-1}$ and $v_1 \mathrel I e_1$ and $v_5 \mathrel I e_7$
  (see \cref{fig:sequence one rule}).
  Only one rule can be applied to $\hg H$ which is an application of
  node-domination to erase $v_1$.
  Then, the hypergraph obtained has a only one rule that can be applied,
  the edge-domination rule that removes $e_2$. We then continue along similar
  lines, starting with $v_3$.
  This example shows that applying a rule can allow to apply another one.
  Moreover, in this sequence of rules, the node-domination and the
  edge-domination rules alternate. Notice that we can generalise this example to a hypergraph
  of an arbitrary length making the sequence of rules arbitrary long, and making
  the number of alternations arbitrarily large.
\end{example}

\begin{figure}[t]
\centering
\begin{minipage}[b]{.44\linewidth}
\begin{subfigure}[b]{\linewidth}
\begin{center}
\begin{subfigure}[b]{\linewidth}
\begin{center}
\begin{tikzpicture}
  \coordinate (c1) at (\gridstep,0);
  \coordinate (c2) at (0,0);
  \coordinate (c3) at (-1.5,-1.5);
  \coordinate (c4) at (2.5,-1.5);
  \setcounter{nodecounter}{0}
  \nodeautocol{red};
  \nodeautocol{blue};
  \nodel{3}{1}{c3}{};
  \nodel{4}{2}{c4}{};
  \twohyperedge{c1}{c2}{\larcrad};
  \twohyperedge{c3}{c3}{\larcrad};
  \twohyperedge{c4}{c4}{\larcrad};
  \node (5) at ([xshift=2 cm] c3) {$\equiv$};
  \draw[->, >=latex, blue] ([shift={(-135:.5)}] c2) to node[below right]{node} ([shift={(45:.5)}]c3);
  \draw[->, >=latex, red] ([shift={(-45:.5)}] c1) to node[below left]{node} ([shift={(135:.5)}]c4);
\end{tikzpicture}
\end{center}
\caption{``Twin'' nodes that lead to two distinct (but isomorphic) hypergraphs.}
\label{fig:twin nodes}
\end{subfigure}

\vspace{2em}
  
\begin{tikzpicture}
  \coordinate (c1) at (\gridstep,0);
  \coordinate (c2) at (0,0);
  \coordinate (c3) at (0,-\gridstep);
  \coordinate (c4) at ([shift={(-45:3)}] c1);
  \coordinate (c5) at ([shift={(-45:3)}] c2);
  \coordinate (c7) at ([shift={(-135:3)}] c1);
  \coordinate (c8) at ([shift={(-135:3)}] c2);
  \coordinate (c9) at ([shift={(-135:3)}] c3);
  \coordinate (c10) at ([yshift= -3cm] c4);
  \coordinate (c11) at ([yshift= -3cm] c5);
  \coordinate (c13) at ([yshift= -3cm] c7);
  \coordinate (c15) at ([yshift= -3cm] c9);
  \setcounter{nodecounter}{0}
  \nodeauto;
  \nodeauto;
  \nodeautocol{red};
  \nodel{4}{1}{c4}{};
  \nodel{5}{2}{c5}{};
  \nodel{7}{1}{c7}{};
  \nodel{8}{2}{c8}{blue};
  \nodel{9}{3}{c9}{};
  \nodel{10}{1}{c10}{};
  \nodel{11}{2}{c11}{};
  \nodel{13}{1}{c13}{};
  \nodel{15}{3}{c15}{};
  \node (16) at ([shift={(1.62,-.5)}] c13) {$\equiv$};
  \twohyperedge{c1}{c1}{\larcrad};
  \twohyperedgecol{c1}{c2}{\barcrad}{blue};
  \twohyperedgev{c2}{c3}{\larcrad};
  \twohyperedge{c4}{c4}{\larcrad};
  \twohyperedgecol{c4}{c5}{\barcrad}{red};
  \twohyperedge{c5}{c5}{\larcrad};
  \twohyperedge{c7}{c7}{\larcrad};
  \twohyperedgev{c8}{c9}{\larcrad};
  \twohyperedge{c10}{c10}{\larcrad};
  \twohyperedge{c11}{c11}{\larcrad};
  \twohyperedge{c13}{c13}{\larcrad};
  \twohyperedge{c15}{c15}{\larcrad};
  \draw[->, >=latex, red] ([shift={(-45:1)}] c1) to node[right]{node} ([shift={(135:1)}]c4);
  \draw[->, >=latex, blue] ([shift={(-135:1)}] c2) to node[left]{edge} ([shift={(45:1)}]c8);
  \draw[->, >=latex, red] ([shift={(.5\gridstep,-1)}] c5) to node[right]{edge} ([shift={(.5\gridstep,.5)}]c11);
  \draw[->, >=latex, blue] ([shift={(.5\gridstep,-.5)}] c9) to node[left]{node} ([shift={(.5\gridstep,1.5)}]c15);
\end{tikzpicture}
\end{center}
\caption{Two sequences of rewritings that diverge in the first step and converge to isomorphic hypergraphs
  in the second step.}
\label{fig:confluence iso}
\end{subfigure}
\end{minipage}\hfill
\begin{minipage}[b]{.53\linewidth}
\begin{subfigure}[b]{\linewidth}
\begin{center}
  {
\begin{tabular}{rl}    
  &
  \begin{tikzpicture}
  \coordinate (c1) at (5.3*\gridstep,0);
  \coordinate (c2) at (4.3*\gridstep,0);
  \coordinate (c3) at (3.3*\gridstep,0);
  \coordinate (c4) at (2.3*\gridstep,0);
  \coordinate (c5) at (1.3*\gridstep,0);
  \coordinate (c6) at (0,0);
  \coordinate (c7) at (60:1.3*\gridstep);
  \setcounter{nodecounter}{0}
  \nodeautocol{red};
  \nodeauto;
  \nodeauto;
  \nodeauto;
  \nodeauto;
  \nodeauto;
  \nodeauto;
  \twohyperedge{c1}{c2}{\larcrad};
  \twohyperedge{c2}{c3}{\barcrad};
  \twohyperedge{c3}{c4}{\larcrad};
  \twohyperedge{c4}{c5}{\barcrad};
  \twohyperedge{c5}{c6}{\larcrad};
  \trighyperedge{c5}{c6}{c7};
\end{tikzpicture}

\\[0.7em]
\begin{tikzpicture}
  \node (N1) at (0,.25) {$\ton$};
  \draw[draw=none] (N1) circle(\offsetrad);
\end{tikzpicture}
&
\begin{tikzpicture}
  \coordinate (c2) at (4.3*\gridstep,0);
  \coordinate (c3) at (3.3*\gridstep,0);
  \coordinate (c4) at (2.3*\gridstep,0);
  \coordinate (c5) at (1.3*\gridstep,0);
  \coordinate (c6) at (0,0);
  \coordinate (c7) at (60:1.3*\gridstep);
  \setcounter{nodecounter}{1}
  \nodeauto;
  \nodeauto;
  \nodeauto;
  \nodeauto;
  \nodeauto;
  \nodeauto;
  \twohyperedge{c2}{c2}{\larcrad};
  \twohyperedgecol{c2}{c3}{\barcrad}{red};
  \twohyperedge{c3}{c4}{\larcrad};
  \twohyperedge{c4}{c5}{\barcrad};
  \twohyperedge{c5}{c6}{\larcrad};
  \trighyperedge{c5}{c6}{c7};
\end{tikzpicture}

\\[0.7em]
\begin{tikzpicture}
  \node (E2) at (0,.25) {$\toe$};
  \draw[draw=none] (E2) circle(\offsetrad);
\end{tikzpicture}
&
\begin{tikzpicture}
  \coordinate (c2) at (4.3*\gridstep,0);
  \coordinate (c3) at (3.3*\gridstep,0);
  \coordinate (c4) at (2.3*\gridstep,0);
  \coordinate (c5) at (1.3*\gridstep,0);
  \coordinate (c6) at (0,0);
  \coordinate (c7) at (60:1.3*\gridstep);
  \setcounter{nodecounter}{1}
  \nodeauto;
  \nodeautocol{red};
  \nodeauto;
  \nodeauto;
  \nodeauto;
  \nodeauto;
  \twohyperedge{c2}{c2}{\larcrad};
  \twohyperedge{c3}{c4}{\larcrad};
  \twohyperedge{c4}{c5}{\barcrad};
  \twohyperedge{c5}{c6}{\larcrad};
  \trighyperedge{c5}{c6}{c7};
\end{tikzpicture}

\\[0.7em]
\begin{tikzpicture}
  \node (N3) at (0,.25) {$\ton$};
  \draw[draw=none] (N3) circle(\offsetrad);
\end{tikzpicture}
&
\begin{tikzpicture}
  \coordinate (c2) at (4.3*\gridstep,0);
  \coordinate (c4) at (2.3*\gridstep,0);
  \coordinate (c5) at (1.3*\gridstep,0);
  \coordinate (c6) at (0,0);
  \coordinate (c7) at (60:1.3*\gridstep);
  \nodelc{2}{};
  \setcounter{nodecounter}{3}
  \nodeauto;
  \nodeauto;
  \nodeauto;
  \nodeauto;
  \twohyperedge{c2}{c2}{\larcrad};
  \twohyperedge{c4}{c4}{\larcrad};
  \twohyperedgecol{c4}{c5}{\barcrad}{red};
  \twohyperedge{c5}{c6}{\larcrad};
  \trighyperedge{c5}{c6}{c7};
\end{tikzpicture}

\\[0.7em]
\begin{tikzpicture}
  \node (E4) at (0,.25) {$\toe$};
  \draw[draw=none] (E4) circle(\offsetrad);
\end{tikzpicture}
&
\begin{tikzpicture}
  \coordinate (c2) at (4.3*\gridstep,0);
  \coordinate (c4) at (2.3*\gridstep,0);
  \coordinate (c5) at (1.3*\gridstep,0);
  \coordinate (c6) at (0,0);
  \coordinate (c7) at (60:1.3*\gridstep);
  \nodelc{2}{};
  \setcounter{nodecounter}{3}
  \nodeauto;
  \nodeauto;
  \nodeauto;
  \nodeauto;
  \twohyperedge{c2}{c2}{\larcrad};
  \twohyperedge{c4}{c4}{\larcrad};
  \twohyperedge{c5}{c6}{\larcrad};
  \trighyperedge{c5}{c6}{c7};
\end{tikzpicture}
\end{tabular}
}
\end{center}
\caption{Sequence of rewritings where only one rule can be applied each time.}
\label{fig:sequence one rule}
\end{subfigure}

\end{minipage}
\caption{Examples of rewrite rules.}
\label{fig:example}
\end{figure}

\paragraph*{Hypergraph isomorphism.}
We now formally define the notion of \emph{hypergraph isomorphism}:

\begin{definition}[Hypergraph isomorphism]
  Let $\hg H = (V,E,I)$ and $\hg H' = (V',E',I')$ be two hypergraphs. We
  say that $\hg H$ is \emph{isomorphic} to $\hg H'$ if there exist two
  bijective functions $f \colon V' \to V$ and $g \colon E' \to E$ which satisfy the following:
  for every $v \in V'$ and $e \in E'$, we have $v \mathrel{I'} e$ if and only if we have
  $f(v) \mathrel{I} g(e)$.
\end{definition}

It is clear that hypergraph isomorphism is an equivalence relation; we
denote it by~$\iso$. For a hypergraph $\hg H$, we denote its equivalence class
under hypergraph isomorphism by $\chg H$.

\paragraph*{Problem statement.}
In this work, we study the rewrite rules of
edge-domination and node-domination on hypergraphs. 
Our goal is to show that these rules are confluent up to isomorphism. In other
words, our goal is to show \cref{thm:RR confluence informal}, which we
reformulate here:

\begin{theorem}[Confluence under isomorphism]\label{thm:RR confluence under isomorphism}
  Let $\hg H, \hg H_1, \hg H_2$ be three hypergraphs such that $\hg H \to^* \hg
  H_1$ and $\hg H \to^* \hg H_2$. Then there exist two hypergraphs $\hg H_3$
  and $\hg H_4$ such that $\hg H_1 \to^* \hg H_3$, $\hg H_2 \to^* \hg H_4$ and
  $\hg H_3 \iso \hg H_4$.
\end{theorem}

Note that, as in the introduction, this directly implies the following corollary:

\begin{corollary}[Unicity of minimal hypergraphs up to isomorphism]\label{coro:unicity minimal hg}
  Let $\hg H$ be any hypergraph, and let $\hg H_1$ and $\hg
  H_2$ be two minimal hypergraphs of $\hg
  H$. Then $\hg H_1$ and $\hg H_2$ are isomorphic.
\end{corollary}
\begin{proof}
  By \cref{thm:RR confluence under isomorphism}, there exist two hypergraphs $\hg H_3$
  and $\hg H_4$ such that $\hg H_1 \to^* \hg H_3$, $\hg H_2 \to^* \hg H_4$ and
  $\hg H_3 \iso \hg H_4$. As $\hg H_1$ and $\hg H_2$ are minimal, we must have $\hg H_1 = \hg H_3$
  and $\hg H_2 = \hg H_4$, hence $\hg H_1 \iso \hg H_2$ indeed.
\end{proof}

In the rest of this note, we prove \cref{thm:RR confluence under isomorphism}.

\section{Proof of \cref{thm:RR confluence under isomorphism}}
\label{sec:proof}

As the rewrite rule $\to$ terminates (\cref{claim:RR terminates}), we will use
\emph{Newman's lemma} (stated below as \cref{lem:Newman}) in order to prove
\cref{thm:RR confluence under isomorphism}.
Thanks to this lemma, it will suffice to prove that $\to$ is \emph{locally
confluent}: we will define this property below, and will then prove that it is
obeyed by our rewrite rules (\cref{prop:RR locally confluent}).
Before all this, however, let us explain how our rewrite rules can be lifted
to equivalence classes:

\begin{definition}[Rewrite rules for classes]
  Let $\leadsto$ be a relation on hypergraphs.
  For two hypergraph classes $\chg{H_1}$ and $\chg{H_2}$, we write
  $\chg{H_1} \leadsto \chg{H_2}$ to mean that there exist
  $\hg H_1' \in \chg{H_1}$ and
  $\hg H_2' \in \chg{H_2}$ such that $\hg H_1' \leadsto \hg H_2'$.
\end{definition}

It is easy to see that the applicability of our rules is preserved up to
isomorphism. This clearly implies that we can lift our rewrite rules to
equivalence classes in the following sense:

\begin{claim}
  \label{claim:CR lift classes}
  Let $\leadsto$ be one of our rewrite rules, i.e. $\leadsto \in \{\toe, \ton, \to\}$.
  For two hypergraph classes $\chg{H_1}$ and $\chg{H_2}$, if $\chg{H_1} \leadsto \chg{H_2}$
  then for all $\hg H_1 \in \chg{H_1}$, there exist $\hg H_2 \in \chg{H_2}$ such that
  $\hg H_1 \leadsto \hg H_2$.
\end{claim}

We are now ready to introduce local confluence for hypergraph classes:

\begin{definition}[Local confluence]\label{def:RR local confluence}
  The rewrite rule $\to$ is
  \emph{locally confluent for equivalence classes} if for all $\chg H$, $\chg{H_1}$, $\chg{H_2}$ such that
  $\chg{H_1}$ and $\chg{H_2}$ are obtained from $\chg H$ in one step
  then there exists $\chg{H_3}$ such
  that $\chg{H_3}$ can be obtained by applying sequences of rewrite rules from $\chg{H_1}$ and
  $\chg{H_2}$. Formally,
  if for all $\chg H$, $\chg{H_1}$ and $\chg{H_2}$ such that $\chg H \to \chg{H_1}$ and $\chg H \to \chg{H_2}$
  then there exists $\chg{H_3}$ such that $\chg{H_1} \to^* \chg{H_3}$ and $\chg{H_2} \to^* \chg{H_3}$.
\end{definition}

Now, we introduce Newman's lemma adapted to our rewrite rule $\to$ on
equivalence classes:

\begin{lemma}[Newman's lemma for $\to$ \cite{biblio:N42, biblio:K85}]\label{lem:Newman}
  If the rewrite rule $\to$ terminates and is locally confluent, then it is confluent.
\end{lemma}

Hence, all that remains to show is the local confluence of $\to$:

\begin{proposition}\label{prop:RR locally confluent}
  The rewrite rule $\to$ is locally confluent for equivalence classes. 
\end{proposition}

The rest of this section is devoted to the proof of this result, which 
suffices to show \cref{thm:RR confluence under isomorphism}. As it turns out, we
will prove a stronger property than local confluence, in that we will only need
at most one single rule application\footnote{This
is close to what is sometimes called the \emph{diamond property}, which allows
exactly one rule application.} to reach isomorphic hypergraphs (whereas the
definition of local confluence would allow an arbitrarily long sequence).

\begin{proof}[Proof of \cref{prop:RR locally confluent}:] Let $\chg H$, $\chg{H_1}$ and
$\chg{H_2}$ be three equivalence classes such that $\chg{H_1}$ and $\chg{H_2}$ are
obtained by one rewrite rule application from $\chg H$, i.e.,
$\chg H \to \chg{H_1}$ and $\chg H \to \chg{H_2}$. We can assume that $\chg{H_1} \neq \chg{H_2}$,
as otherwise the result immediately holds.
Let $\hg H = (V,E,I)$,
$\hg H_1 = (V_1,E_1,I_1)$ and $\hg H_2 = (V_2,E_2,I_2)$ be representatives
of $\chg H$, $\chg{H_1}$ and $\chg{H_2}$, respectively, so that $\hg H \to \hg H_1$ and
$\hg H \to \hg H_2$.
Note that we can choose the same hypergraph $\hg H$ as a representative of $\chg H$ for both rewrite rules according to
\cref{claim:CR lift classes}.
We do a case disjunction on the rewrite rule that was
applied in the two cases:

\begin{enumerate}[label=(\roman*)]
\item
  \label{proof:double edge-domination}
  If $\hg H \toe \hg H_1$ and $\hg H \toe \hg H_2$:\\
  Let $e_1, e_1', e_2, e_2' \in E$ be the witnessing hyperedges for the two rule
    applications: 
    we have that $e_1' \neq e_2'$,
    that  $V(e_1) \subseteq V(e_1')$, and that $V(e_2) \subseteq V(e_2')$,
    and we have that $\hg H_1$ is obtained from $\hg H$
    by removing $e_1'$ ($V_1 = V$, $E_1 = E \setminus \{e_1'\}$ and
    $I_1 = \{(v,e) \in I \colon e \neq e_1'\}$) and
  likewise $\hg H_2$ is obtained from $\hg H$ by removing $e_2'$.
  Let $\hg H_3 = (V_3,E_3,I_3)$ be the hypergraph obtained from $\hg H$ by removing
  $e_1'$ and $e_2'$, i.e., we set $V_3 \coloneq V$, $E_3 \coloneq E \setminus \{e_1', e_2'\}$
    and $I_3 \coloneq \{(v,e) \in I \colon e \neq e_1' \text{~and~} e \neq e_2'\}$.
  We have several easy sub-cases.
  \begin{itemize}  
  \item If $e_1' \neq e_2$ and $e_2' \neq e_1$ then, as we already have that
    $e_1' \neq e_2'$,
    we can simply remove $e_2'$ 
    in $\hg H_1$
    (because $e_2$ is still there and we still have $V(e_2) \subseteq V(e'_2)$) to yield $\hg H_3$. 
    For the same reason, 
    we can obtain $\hg H_3$ from $\hg H_2$ by removing $e'_1$. 
Therefore $\hg H_1 \toe^* \hg H_3$ and $\hg H_2 \toe^* \hg H_3$, so local confluence is verified.
  \item 
    If $e_1 \neq e_2'$ and $e_2 = e_1'$, then we have $V(e_1) \subseteq V(e_1') \subseteq V(e_2')$
    and once again, we can remove both independently using $e_1$ as a witness, which satisfies local confluence.
  \item Symmetrically, if $e_1 = e_2'$ and $e_2 \neq e_1'$, we have the same result.
  
  \item Otherwise, $e_1 = e_2'$ and $e_2 = e_1'$, and we have $V(e_1) \subseteq
    V(e_1') = V(e_2)$ and $V(e_2) \subseteq V(e'_2) = V(e_1)$, so that $V(e_1) =
    V(e_2)$. We already have that $e_1' \neq e_2'$, so that $e_1 \neq e_2$.
      Here, once we have removed one of the two edges,
      it may not be possible to remove the second one, because the witness
      justifying the applicability of the second edge-removal rule might not exist
      anymore. Nevertheless, in fact we immediately have that $\hg H_1$ is isomorphic to $\hg H_2$, via the bijections 
      $f \colon V_2 \to V_1 = \mathit{id}$ and $g \colon E_2 \to E_1$ defined by 
      $g_{|E \setminus \{e_1, e_2\}} = \mathit{id}$ and $g(e_2) = e_1$. As $\hg H_1$ and $\hg H_2$ are isomorphic, local confluence for equivalence classes holds.
  \end{itemize}
  
\item 
    \label{proof:mixed-domination}
  If $\hg H \ton \hg H_1$ and $\hg H \toe \hg H_2$:\\
  Let $v_1,v_1' \in V$ and $e_2, e_2' \in E$ such that $v_1 \neq v_1'$, $E(v_1) \subseteq E(v_1')$,
  $e_2 \neq e_2'$ and $V(e_2) \subseteq V(e_2')$.
  We have that $\hg H_1$ is obtained by removing $v_1$ ($V_1 = V \setminus
  \{v_1\}$, $E_1 = E$
  and $I_1 = \{(v,e) \in I \colon v\neq v_1\}$)
  and that $\hg H_2$ is obtained by removing $e_2'$.
  Let $\hg H_3 = (V_3,E_3,I_3)$ be the hypergraph obtained from $\hg H$ in which
    we removed $v_1$ and $e_2'$: formally,
  $V_3 \coloneq V \setminus \{v_1\}$, $E_3 \coloneq E \setminus \{e_2'\}$, and
    $I_3 \coloneq \{(v,e) \in I \colon e \neq e_2' \text{~and~} v \neq v_1\}$.
  We need to show that $\hg H_1 \toe \hg H_3$ and $\hg H_2 \ton \hg H_3$.

  In $\hg H_2$,
  in all cases if we removed $e_2'$ from $E(v_1)$ then we also erased it from $E(v_1')$.
  More formally, since $E(v_1) \subseteq E(v_1')$, and since we have $E_2(v_1) = E(v_1) \setminus \{e_2'\}$ and
  the same for $E_2(v_1')$, then we have $E_2(v_1) \subseteq E_2(v_1')$.
  This means we can apply the node-domination rule to $\hg H_2$ on~$v_1$ and
  $v_1'$ and obtain $\hg H_3$, hence $\hg H_2 \ton \hg H_3$.

  Likewise, in $\hg H_1$, the situation is analogous up to exchanging the roles
  of nodes and edges: if we removed $v_1'$ from $E(e_2)$ then we also erased it
  from~$E(e_2')$.
  Hence $\hg H_1 \toe \hg H_3$.

\item If $\hg H \toe \hg H_1$ and $\hg H \ton \hg H_2$:\\
  As we can exchange
  the role of $\chg{H_1}$ and $\chg{H_2}$ without loss of generality, this case
    is symmetric to 
    \cref{proof:mixed-domination}.

\item Else $\hg H \ton \hg H_1$ and $\hg H \ton \hg H_2$:\\
  This case is analogous to \cref{proof:double edge-domination} but exchanging
  the roles of nodes and of edges, considering the node-domination rule
  instead of the edge-domination rule, and reversing all inclusion
  relationships on sets of incident objects.
\end{enumerate}

In all cases, the reasoning is independent of the representatives. We conclude that
$\chg{H_1} \to^* \chg{H_3}$ and $\chg{H_2} \to^* \chg{H_3}$.
\end{proof}

From \cref{prop:RR locally confluent} together with \cref{claim:RR terminates},
using \cref{lem:Newman} we deduce the main result of this note (\cref{thm:RR confluence under
isomorphism}).

\pagebreak

\bibliographystyle{plain}
\bibliography{references}

\vfill
\doclicenseThis

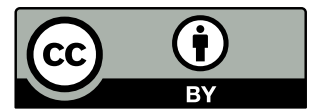
\end{document}